\documentclass[a4paper]{article}
\pdfoutput=1

\usepackage[numbers,sort&compress]{natbib}
\usepackage{enumitem}
\usepackage{comment}
\usepackage{amsmath, amsthm, amssymb}
\usepackage{thmtools}
\usepackage{mathtools}

\newtheorem{proposition}{Proposition}
\newtheorem{theorem}{Theorem}
\newtheorem{definition}{Definition}

\newtheorem{corollary}{Corollary}
\newtheorem{lemma}{Lemma}

\DeclareMathOperator*{\argmin}{argmin}

\newcommand{\N}{\mathbb{N}}

\newcommand{\R}{\mathbb{R}}

\renewcommand{\Pr}{\mathop{\mathrm{Pr}}}

\newcommand{\parens}[1]{{\left({#1}\right)}}
\newcommand{\brackets}[1]{{\left[{#1}\right]}}

\newcommand{\floor}[1]{{\left\lfloor{#1}\right\rfloor}}
\newcommand{\ceil}[1]{{\left\lceil{#1}\right\rceil}}
\newcommand{\poly}{\mathrm{poly}}

\newcommand{\zo}{\{0, 1\}}
\newcommand{\eps}{\varepsilon}

\def\FullBox{\hbox{\vrule width 8pt height 8pt depth 0pt}}
\def\qed{\ifmmode\qquad\FullBox\else{\unskip\nobreak\hfil
\penalty50\hskip1em\null\nobreak\hfil\FullBox
\parfillskip=0pt\finalhyphendemerits=0\endgraf}\fi}

\usepackage{float}
\usepackage{hyperref}
\usepackage[capitalise]{cleveref}
\crefname{proposition}{Proposition}{Propositions}
\crefname{theorem}{Theorem}{Theorems}
\crefname{definition}{Definition}{Definitions}
\crefname{conjecture}{Conjecture}{Conjectures}
\crefname{corollary}{Corollary}{Corollaries}
\crefname{lemma}{Lemma}{Lemmas}
\crefname{claim}{Claim}{Claims}
\crefname{fact}{Fact}{Facts}
\crefname{example}{Example}{Examples}
\crefname{remark}{Remark}{Remarks}
\crefname{table}{Table}{Tables}
\crefname{section}{Section}{Sections}
\crefname{subsection}{Subsection}{Subsections}
\crefname{algorithm}{Algorithm}{Algorithms}
\usepackage[margin=1.0in]{geometry}

\usepackage[noend]{algpseudocode}
\usepackage{algorithm}
\usepackage{multirow}

\title{Succinct Oblivious RAM\thanks{A preliminary version of this paper
appeared in the {\it Proceedings of the 35th Symposium on Theoretical Aspects
of Computer Science}, 2018, pp.  52:1--52:16.}}
\author{
  Taku Onodera\thanks{Human Genome Center, Institute of Medical Science, the
  University of Tokyo} \\
  \texttt{tk-ono@hgc.jp}
  \and
  Tetsuo Shibuya\footnotemark[2] \\
  \texttt{tshibuya@hgc.jp}}
\date{}

\begin{document}

\maketitle

\begin{abstract}
Reducing the database space overhead is critical in big-data processing.
In this paper, we revisit oblivious RAM (ORAM) using big-data standard for the
database space overhead.

ORAM is a cryptographic primitive that enables users to perform arbitrary
database accesses without revealing the access pattern to the server.
It is particularly important today since cloud services become increasingly
common making it necessary to protect users' private information from database
access pattern analyses.
Previous ORAM studies focused mostly on reducing the access overhead.
Consequently, the access overhead of the state-of-the-art ORAM constructions
is almost at practical levels in certain application scenarios such as secure
processors.
On the other hand, most existing ORAM constructions require $(1+\Theta(1))n$
(say, $10n$) bits of server space where $n$ is the database size.
Though such space complexity is often considered to be ``optimal'', overhead
such as $10 \times$ is prohibitive for big-data applications in practice.

We propose ORAM constructions that take only $(1+o(1))n$ bits of server space
while maintaining state-of-the-art performance in terms of the access overhead
and the user space.
We also give non-asymptotic analyses and simulation results which indicate that
the proposed ORAM constructions are practically effective.
\end{abstract}

\section{Introduction} \label{section: intro}

\emph{Oblivious RAM (ORAM)} is a cryptographic primitive that enables users to
access a database on a server without revealing the access pattern to the
server.\footnotemark
\footnotetext{In the original paper, an ORAM is defined to be a random access
machine for which the memory access pattern is independent of the
input~\cite{Goldreich87}. Our use of the word ORAM follows the convention of
some subsequent work, e.g., \cite{Ren15,Wang15,Devadas16}.}
Although originally introduced in the context of software
protection~\cite{Goldreich87}, ORAM is directly relevant to the present cloud
computing scenarios.

In the previous studies on ORAM, researchers focused mainly on reducing the
access bandwidth cost, a performance measure used as a proxy of the access
time.
This is because even the current most state-of-the-art ORAM constructions have
two or three orders of magnitude larger bandwidth cost than the ordinary
(non-secure) accesses.
However, in certain settings, the ORAM access is already rather efficient.
For example, Maas et al. proposed PHANTOM~\cite{Maas13}, an ORAM-based secure
processor, and reported that if PHANTOM is deployed on the server, SQLite
queries can be performed without revealing the access pattern at the cost of
1.2--6$\times$ slowdown compared to non-secure SQLite queries.
In such cases, it is reasonable to pay more attention to performance measures
other than the access speed.

In particular, the \emph{server space usage} is a very important performance
measure for big-data applications.
First, there are applications where the amount of data is virtually unbounded,
and thus the limit of the available space defines the limit of the analyses.
Second, due to the cache effect, small memory usage often leads to faster
computation.
Third, space costs money, especially in a cloud computing server.
The second and the third points are especially relevant if the data is meant to
be stored in the main memory (by default), which is exactly the case in ORAM
application scenarios such as PHANTOM.

In most modern ORAM constructions, if the size of the original database is $n$
bits, the amount of the space required by the server is $n+\Theta(n)$ bits.
In this paper, we investigate the possibility of ORAM constructions that need
only $n+o(n)$ bits of server space.
We call such ORAM constructions \emph{succinct}.
This space efficiency formalization is widely used in the field of
\emph{succinct data structures} and has proved to be useful to design
practically relevant space-efficient data structures in theoretically clean
ways.

The main difficulty to achieve succinctness is that most existing ORAM
construction approaches rely on the use of linear amount of "dummy" data.
The situation is similar to conventional hash tables, which need extra space
linear to the stored keys size.
Although it seems possible to reduce the constant factor of the extra space to
some extent, it is not at all trivial if one can achieve sublinear extra space
maintaining the state-of-the-art performance in other aspects such as access
bandwidth and user space usage.

\paragraph{Results.}

\cref{table: theoretical performance} shows the performance comparison of the
proposed methods and the existing methods.
Our first construction takes
$n(1+\Theta(\frac{\log{n}}{B}+\frac{g(n)}{f_1(n)/\log{n}}))$-bit server space
where $n$ is the database size, $f_1(\cdot)$ is an arbitrary function such that
$f_1(n)=\omega(\log{n})$ and $O(\log^2{n})$, $g(\cdot)$ is an arbitrary
function such that $g(n)=\omega(1)$ and $o(\sqrt{f_1(n)/\log{n}})$, and $B$ is
the size of a \emph{block}, the unit of communication between the user and the
server.
The bandwidth blowup is $O(\log^2{n})$ and the user space is $O(f_1(n))$
blocks.
Our second construction achieves
$n(1+\Theta(\frac{\log{n}}{B}+\frac{\log\log{n}}{f_2(n)}))$-bit server space,
$O(\log^2{n})$-bandwidth blowup and $O(f_2(n)+R(n))$-user space where
$f_2(\cdot)$ is an arbitrary function such that $f_2(n) = \omega(\log\log{n})$
and $O(\log^2{n})$, $R(\cdot)$ is an arbitrary function such that $R(n) =
\omega(\log{n})$.

For example, suppose $B = \lg^2{n}$, $R = \lg{n}\lg\lg{n}$, $f_1(n) = f_2(n) =
\lg{n}\lg\lg{n}$ and $g(n) = \lg\lg\lg{n}$.
Then, the user space of each of our first and the second constructions is
$O(\log{n}\log\log{n})$ and the server space is
$n(1+\Theta(\frac{\log\log\log{n}}{\log\log{n}}))$ (resp.
$n(1+\Theta(\frac{1}{\log{n}}))$) bits in the first (resp. second)
construction.

The second construction has better theoretical performance than the first one.
However, in practice, with some parameter settings, the first construction also
works comparably well as the second construction depending on which performance
measure one cares (See \cref{section: non-asymptotic}).
The first construction is also the basis of the second construction.

If $B = \omega(\log{n})$, Goldreich's construction~\cite{Goldreich87} and our
constructions are succinct.
(Each of these methods works as long as $B \ge c\lg{n}$ for $c$ around 3.)
The assumption $B = \omega(\log{n})$ is justified as follows.
Stefanov et al.~\cite{Stefanov13} mentioned that the typical block size is
64--256 KB (resp. from 128B to 4KB) in cloud computing scenario (resp. software
protection scenario).
Even $B \ge \lg^{1.5}{n}$ holds if $n \le 2^{6501}$ (resp. $n \le 2^{97}$) in
cloud computing (resp. software protection) scenario with moderate block size
of 64KB (resp. 128B).

We achieved exponentially smaller bandwidth blowup compared to Goldreich's
construction~\cite{Goldreich87}, which is the only preceding non-trivial
succinct ORAM construction.

The bandwidth blowup of our constructions is smaller or equal to other
non-succinct constructions except the construction of Kushilevitz et
al.~\cite{Kushilevitz12}, the Onion ORAM~\cite{Devadas16} and the so called SSS
construction~\cite{Stefanov12}.
The construction of Kushilevitz et al. (and every other construction that is
listed above it in \cref{table: theoretical performance}) is based on a very
expensive procedure called oblivious sorting and the constant factor hidden in
the asymptotic notation of the bandwidth cost is prohibitively large.
The Onion ORAM achieves $O(1)$-bandwidth blowup but it requires several
assumptions.
First, the Onion ORAM requires the server to perform some computation, e.g.,
homomorphic encryption evaluation.
(In every other construction in \cref{table: theoretical performance}, the
server suffices to respond to read/write requests.)
It also requires a computational assumption (decisional composite residuosity
assumption or learning with errors assumption), and larger block size
($B=\widetilde{\omega}(\log^2{n})$ to $\widetilde{\omega}(\log^6{n})$ depending
on the exact construction, where $\widetilde\omega(\cdot)$ hides a polyloglog
factor).
The SSS construction takes $cn$-bit user space where $c \ll 1$.
This method is effective for ordinary cloud computing setting but the user
space is too large for secure processor setting --- the PHANTOM-like
applications where server space efficiency is more important.

\begin{table}
  \centering
  \caption{Comparison of theoretical performances.
  Bandwidth blowup is the number of blocks required to be communicated for
  accessing one block of data.
  User space includes the temporary space needed during access procedures.
  $n$ is the database size in bits and $B$ is the block size in bits.
  $B$ must satisfy $B \ge c_1\lg{n}$ and $B = O(n^{c_2})$ for constants
  $c_1>1$, $0<c_2<1$.
  Typically, $c_1$ is around $3$.
  $f_1(\cdot)$ is an arbitrary function such that $f_1(n) = \omega(\log{n})$
  and $O(\log^2{n})$.
  $f_2(\cdot)$ is an arbitrary function such that $f_2(n) =
  \omega(\log\log{n})$ and $O(\log^2{n})$.
  $R(\cdot)$ is an arbitrary function such that $R(n) = \omega(\log{n})$.
  $g(\cdot)$ is an arbitrary function such that $g(n) = \omega(1)$ and
  $o(\sqrt{f_1(n)/\log{n}})$.
  Bounds with $\dagger$ are amortized.
  The method in~\cite{Devadas16} requires additional assumptions.
  The user space bound of the method in~\cite{Stefanov12} has a constant factor
  $\ll 1$.}
  \label{table: theoretical performance}
  \begin{tabular}{c | c c c}
    & Server space (\#bits)
    & \begin{tabular}[x]{@{}c@{}}Bandwidth\\blowup\end{tabular}
    & \begin{tabular}[x]{@{}c@{}}User space\\(\#block)\end{tabular} \\
  
    \hline
  
    Goldreich~\cite{Goldreich87} &
  	$n\parens{1+\Theta\parens{\frac{\log{n}}{B}+\frac{1}{\sqrt{n}}}}$ &
  	$O(\sqrt{n}\log{n})^\dagger$ &
    $O(1)$ \\
  
    Ostrovsky~\cite{Ostrovsky90} &
    $O(n\log{n})$ &
    $O(\log^3{n})^\dagger$ &
    $O(1)$ \\
  
    Ostrovsky, Shoup~\cite{Ostrovsky97} &
    $n(1+\Theta(1))$ &
  	$O(\sqrt{n}\log{n})$ &
    $O(1)$ \\
  
    Ostrovsky, Shoup~\cite{Ostrovsky97} &
    $O(n\log{n})$ &
    $O(\log^3{n})$ &
    $O(1)$ \\
  
    Goodrich, Mitzenmacher~\cite{Goodrich11} &
    $n(1+\Theta(1))$ &
    $O(\log^2{n})^\dagger$ &
    $O(1)$ \\
  
    Kushilevitz, Lu, Ostrovsky~\cite{Kushilevitz12} &
    $n(1+\Theta(1))$ &
    $O(\frac{\log^2{n}}{\log\log{n}})$ &
  	$O(1)$ \\
  
    Stefanov, Shi, Song~\cite{Stefanov12} &
    $n(1+\Theta(1))$ &
    $O(\log{n})$ &
    $O(n)$ \\
  
    Stefanov et al.~\cite{Stefanov13} &
    $n(1+\Theta(1))$ &
    $O(\log^2{n})$ &
    $O(R(n))$ \\
  
    Devadas et al.~\cite{Devadas16} &
    $n(1+\Theta(1))$ &
    $O(1)$ &
    $O(1)$ \\
  
    \hline
  
    Our result (\cref{main theorem}) &
    $n\parens{1+\Theta\parens{\frac{\log{n}}{B}+\frac{g(n)}{f_1(n)/\log{n}}}}$ &
  	$O(\log^2{n})$ &
    $O(f_1(n))$ \\
  
    Our result (\cref{second main theorem}) &
    $n\parens{1+\Theta\parens{\frac{\log{n}}{B}+\frac{\log\log{n}}{f_2(n)}}}$ &
  	$O(\log^2{n})$ &
    $O(f_2(n)+R(n))$
  \end{tabular}
\end{table}

\paragraph{Possible applications.}

There are several ORAM application scenarios with different requirements.
Our methods are particularly relevant to \emph{secure processor} scenario.
In this scenario, it is assumed that a special processor under the control of
the user is available in a remote server and the adversary cannot observe the
activities inside the processor.
The cloud service user sends a piece of code to the trusted processor, which,
in turn, executes the code on the server.
The communication between the cloud service user and the secure processor is
protected by private key encryption.
ORAM is implemented inside of the trusted processor using FPGA and it hides the
processor's access pattern to the main memory on the server.
After executing the code, the secure processor may return the (encrypted)
output to the cloud service user.
One of the main advantages of this approach over the conventional ORAM
application, in which the cloud service user locally executes ORAM, is that
ORAM bandwidth blowup applies to the relatively cheap processor--memory
communication rather than the costly over-network communication.
Note that, with the ORAM user-server terminology, the secure processor (resp.
the main memory) is the user (resp. the server).

In secure processor scenario,
\begin{itemize}[noitemsep]
  \item the user space is very limited, e.g., 6MB;
  \item The server usually does not perform complex computation;
  \item Simple ORAM algorithms are desirable for hardware implementation;
  \item The server space is much larger than the user space but there is some
  noticeable limit.
  The server can use disks if needed but it greatly slows down accesses.
\end{itemize}
In most existing secure processor systems, the Path ORAM~\cite{Stefanov13} or
its close variants are used~\cite{Fletcher12,Maas13,Ren13,Fletcher15}.
Indeed, the Path ORAM satisfies the first three requirements above.
However, it does not capture the last point.
For example, suppose 128GB database is stored in the Path ORAM.
If the block size is 128B, it takes about 10G blocks, i.e., 1.28TB (to ensure
rigorous security).
Then, each ORAM access procedure takes about 31$\mu$s assuming each memory
access takes 100ns.
If half of the 10G blocks are stored in the main memory and the other half is
stored in the disk, due to the randomized access pattern of the Path ORAM,
almost every ORAM access procedure ends up a disk seek, which takes
milliseconds order time.
In such cases, it is reasonable to use another ORAM construction that takes,
say, half the space of the Path ORAM even though it requires twice as many
memory accesses.

\paragraph{Tree-based ORAM.}

Our ORAM constructions are tree-based.
In a typical tree-based ORAM construction, $N$ blocks are stored in a complete
binary tree with $N$ leaves on the server.
Each node of the tree can store up to $Z$ blocks where $Z$ is a constant.
Each block is assigned a position label, which is an integer chosen uniformly
at random from $[N]$.
A block with position label $i$ must be stored at some node on the path from
the root to the $i$-th leaf.
This framework was introduced by Shi et al.~\cite{Shi11} and used in many
subsequent studies~\cite{Stefanov13, Gentry13, Ren13, Chung14, Devadas16}.

Consider a particular block $b$.
As the user continuously issues access requests, $b$ moves around the tree in
roughly the following manner.
First, when the user issues an access request to $b$, $b$ is picked out of the
tree and given a new uniformly random position label.
Then, $b$ is inserted into the tree from the root.
If the user issues an access request to another block, then, with some
probability, $b$ will move down the path to the leaf indicated by its position
label.
If the next node on the path is full, $b$ must wait for the blocks ``ahead'' to
move down.
If the pace at which the blocks move down the tree cannot keep up with the pace
at which blocks are picked out and reinserted from the root, then, some blocks
will not be able to reenter the tree.
If such ``congestion'' occurs, the user must maintain the overflown blocks
locally.

Note that most space in the tree is wasted: there are $2N-1$ nodes in the tree,
each with capacity $Z$, whereas there are only $N$ blocks.
Thus, to save server space, it is desirable to make the tree more compact, for
example, by reducing $Z$.
However, to maintain a low probability of ``congestion'', it is desirable to
make the tree larger, for example, by increasing $Z$.
To construct a succinct tree-based ORAM, we need to satisfy these conflicting
demands simultaneously.

\paragraph{Our ideas.}

One of our key ideas is the following two-stage tree layout.
We first change the tree to a complete binary tree with $N/\lg^{1.4}{N}$ leaves
(assume this is a power of 2).
In addition, we set the capacity of each leaf node to $\lg^{1.4}{N} +
\lg^{1.3}{N}$ while keeping the capacity of each internal node at $Z$.
The total size of the leaf nodes is then $N + N/\lg^{0.1}N$, and the total size
of all tree nodes except the leaves is $\Theta(N/\lg^{1.4}{N})$.
Thus, the total size of the entire tree is $N+o(N)$.
We choose each position label from $[N/\lg^{1.4}{N}]$.

To see why blocks can flow around in this tree without much congestion, suppose
that the user inserts each block directly into the leaf node pointed to by the
block's position label.
Clearly, the loads of leaves in this hypothetical setting dominates the loads
of leaves in the real setting.
Then, the situation would exactly be the same as the ``balls-into-bins''
game~\cite{MitzenmacherUpfal} with $N$ balls and $N/\lg^{1.4}{N}$ bins.
In particular, the number of blocks stored in each leaf node is
$\log^{1.2}{N}+\Theta(\log^{0.6}{N})$ with high probability.
Thus, every leaf node has sufficient capacity to store all of its assigned
blocks.

Furthermore, the blocks in the internal nodes flow as smoothly as in the
original non-succinct ORAM construction since we did not modify that part.
Therefore, the blocks flow without much congestion throughout the tree.
This is the idea behind the first construction (\cref{main theorem}).

Another key idea follows naturally from the above argument, specifically from
the connection to the balls-into-bins game.
A remarkable phenomenon known as ``the power of two choices'' states that, in
the balls-into-bins game, if one chooses two bins uniformly and independently
for each ball, and throws the ball into the least loaded bin, the bin loads
will be distributed much more tightly around the mean than they are in the
one-choice game~\cite{Azar99,Berenbrink00,MitzenmacherUpfal}.
The maximum bin load corresponds to the leaf node size in tree-based ORAM
constructions.
Thus, the size of the tree can be further decreased by using the two-choice
strategy to assign the position labels.
This is the idea behind the second construction (\cref{second main theorem}).

We note that the current paper is the first to apply the power of two choices
to tree-based ORAM.
(Some non-tree-based constructions~\cite{Pinkas10, Goodrich11, Kushilevitz12}
use the two choices idea in the form of cuckoo hashing~\cite{Pagh04}.)
Moreover, the resulting algorithms keep the simplicity of the Path
ORAM~\cite{Stefanov13}, which is a highly valuable asset in the relevant
application scenario as mentioned above.
As for the analysis, the existing stash size analyses~\cite{Stefanov13,Ren13}
do not seem to work with parameter regimes required for succinctness.
We will give a different proof route (though it still heavily borrows
from~\cite{Stefanov13,Ren13}).

\paragraph{Our contributions.}

Our contributions in the current paper are as follows:
\begin{itemize}[noitemsep]
  \item We introduce the notion of succinct oblivious RAM.
  This is a promising first step to systematically design ORAM constructions
  with small server space usage;
  \item We propose two succinct ORAM constructions.
  Not only being succinct, these constructions exhibit state-of-the-art
  performance in terms of the bandwidth blowup.
  The methods are simple and easy to implement;
  \item We also give non-asymptotic bounds and simulation results which
  indicate that the proposed methods are practically effective.
\end{itemize}

\subsection{Related Work}

In the field of succinct data structures~\cite{Jacobson88,Jacobson89}, the goal
is to represent an object such as a string~\cite{Munro01-1, Sadakane02,
Grossi03, Ferragina05-1, Golynski06, Sadakane06, Jansson07, Ferragina07, Hon14,
Navarro14-2} or a tree~\cite{Clark96, Munro01-2, Raman03, Benoit05,
Ferragina05-2, Navarro14} in such a way that a) only $OPT+o(OPT)$ bits are
required, and b) relevant queries such as random access or substring search are
efficiently supported.
Here, $OPT$ is the information theoretic optimum, i.e., the minimum number of
bits needed to represent the object.

The current study is related to succinct data structures in the following way.
Suppose a remote server hosts a database that is implemented by a succinct data
structure, and a user wishes to access the database without revealing the
access pattern to the server.
The user, of course, can apply any existing ORAM constructions.
However, if ORAM increases the database size by some constant factor, it
destroys the $OPT+o(OPT)$ bound guaranteed by the succinct data structure.
One can apply the succinct ORAM constructions proposed in this paper to hide
succinct data structure access pattern on a remote storage device without
harming the theoretical guarantee on the data structure size.

\subsection{Organization of the Paper}

In \cref{section: prelim}, we introduce basic notions that will be used in
later sections.
We describe our first succinct ORAM construction (encapsulated in \cref{main
theorem}) in \cref{section: succinct oram} and the second construction
(encapsulated in \cref{second main theorem}) in \cref{section: succincter
oram}.
Then, we present non-asymptotic analyses and simulation results in
\cref{section: non-asymptotic}.
We conclude the paper in \cref{section: conclusion}.

\section{Preliminaries} \label{section: prelim}

\subsection{Notations}

We denote the set $\{0, 1, \dots n-1\}$ as $[n]$ for a non-negative integer
$n$.
We write $\lg{x}$ to denote the base-$2$ logarithm of $x$ and $\ln{x}$ to
denote the natural logarithm of $x$.
We write $\log{x}$ to denote the logarithm of $x$ in the context where the base
can be any positive constant.
We write $\poly(n)$ to denote $n^c$ for some constant $c>0$.
A negligible function of $n$ is defined to be a function that is asymptotically
smaller than $1/n^c$ for any constant $c>0$.

\subsection{Oblivious RAM} \label{subsec: ORAM}

\paragraph{Definition.}

Oblivious RAM is defined through the interaction between three parties the
\emph{user}, the \emph{server} and the \emph{oblivious RAM (ORAM) simulator}.
The user wishes to perform random access to the database on the server without
revealing the ``access pattern'' to the server.
Roughly speaking, the ORAM simulator works as a mediator between the user and
the server.
It takes access requests to the database from the user and translates them to
``appropriate'' access requests to the server.
The database on the server can be maintained as some ``data structure'' instead
of the raw form on which the user intends to perform random access and thus,
the access requests to the server need not be (and are not) the same as the
access requests to the database.
The ORAM simulator then, performs random accesses to the server on behalf of
the user (using translated requests), thereby making it impossible for the
server to infer the access patterns to the database even though the accesses to
the server is visible from the server.

Formally, let each of $B$ and $n$ be a positive integer and $N := \ceil{n/B}$.
The value $B$ models the unit of communication and $n$ models the database
size.
We call a chunk of $B$ bits a \emph{block}.
For brevity, we assume $n$ is a multiple of $B$ in the rest of the paper.
A \emph{logical (resp. physical) access request} is a triplet $(\mathrm{op},
\mathrm{addr}, \mathrm{val})$, where $\mathrm{op} \in
\{\mathrm{read},\mathrm{write}\}$, $\mathrm{addr} \in [N]$ (resp.
$\mathrm{addr} \in \N$), $\mathrm{val} \in \zo^B$.
The user sends logical access requests to the ORAM simulator and receives a
block for each request.
The server receives physical access requests from the ORAM simulator and
returns a block for each request in the following way: for $(\mathrm{read}, i,
v)$, the server returns $v$ of the most recent request $(\mathrm{write}, i,
v)$.
The ORAM simulator takes a sequence of logical access requests from the user
and for each logical access request, it makes a sequence of physical access
requests to the server receiving a returned block for each of them, and returns
a block to the user.
The ORAM simulator is possibly stateful and probabilistic.
It must respond to logical access requests online and must satisfy the
following conditions:

\begin{description}
  \item[Correctness] The ORAM simulator is correct if and only if, for a
  logical access request with $\mathrm{addr}=i$, it returns $v$ of the previous
  and most recent logical access request $(\mathrm{write},i,v)$;\footnotemark
  \footnotetext{We use the convention that not only read but also write
  requests have return values.}
  \item[Security] The ORAM simulator is computationally (resp. information
  theoretically) secure if and only if, for any logical access request
  sequences of the same length, the distributions of the $\mathrm{addr}$ values
  of the resulting physical access requests are computationally (resp.
  information theoretically) indistinguishable.
\end{description}

An ORAM construction is an ORAM simulator implementation.
We have distinguished the user from the ORAM simulator for exposition but in
practice, an ORAM simulator is a program run by the user.
Thus, we do not distinguish them in the rest of the paper.

\paragraph{Encryption.}

In the ORAM constructions considered in this paper, the user holds a symmetric
cipher key and every block is encrypted when it is stored on the server.
Encryption can increase the database size.
Theoretically, we can bound the space overhead due to encryption to
$o(1)$-factor.
For example, one can encrypt a block $m$ as $(r,m \oplus F(r))$ where $r$ is a
random bits of size $\omega(\log{n})$ and $o(B)$, $F$ is a pseudorandom
function (key is omitted) and $\oplus$ denotes bitwise XOR.
Or, in practice, one can use ``counter mode'' of block cipher, i.e., encrypting
a block $m$ as $(i,m \oplus F(z||i))$ where $F$ is AES, $i$ is the number of
blocks encrypted so far and $z$ is a nonce.
Assuming that we allocate 128 bits to $i$ and the typical block sizes mentioned
in \cref{section: intro}, the additional space is 1/4096--1/16384 (resp.
1/8--1/256) factor of the original database size in cloud computing (resp.
software protection) scenario.
Since the space overhead due to encryption is rather small, we ignore it in the
rest of the paper.

\paragraph{Performance measures.}

The most popular ORAM performance measures include the amount of the space
required by the user/server and the amount of time required for each logical
access.

In most ORAM constructions, the user needs to maintain a small amount of
information locally.
In addition to this, in some constructions, the user temporarily needs to store
more information during the access procedure.
We refer the amount of the space the user temporarily needs during access
procedure as \emph{temporary space usage} and the amount of the space the user
needs even if no access is made as \emph{permanent space usage}.

In this paper, we pay special attention to the server space usage.
In particular, we use the following notion of \emph{succinctness} as a
criterion for ORAM server-space efficiency:
\begin{definition}
If the server space usage of an ORAM construction representing an $n$-bit
database is $n + o(n)$ bits, the ORAM construction is said to be succinct.
\end{definition}

As for the access efficiency, following the previous studies, we use the amount
of communication between the user and the server as a proxy for the access
time.
We define the \emph{bandwidth blowup} of an ORAM construction to be the number
of blocks that needs to be communicated between the user and the server per
logical access.
In other words, the bandwidth blowup is the ratio of communication amount
needed for secure access to communication amount needed for ordinary (insecure)
access.\footnotemark
\footnotetext{The bandwidth blowup is a ratio and does not have a unit.}

\paragraph{Asymptotic behavior of parameters.}

Among the ORAM-related parameters, the original database size $n$ and block
size $B$ are outside of the user's control.
Other parameters, e.g., the metadata size, can be chosen by the user.
We assume that $B$ is a function of $n$ satisfying $B = \omega(\log{n})$.
(See \cref{section: intro} for the justification.)
Thus, after all, $n$ is the only free parameter on which the other parameters
depend.
In all asymptotic statements in this paper, the limit is taken as $n \to
\infty$.

\subsection{Sub-ORAM}

We use an ORAM construction encapsulated into the following proposition as a
blackbox.
Concretely, the Path ORAM~\cite{Stefanov13} suffices.

\begin{proposition} \label{prop: Path ORAM}
Let $n$ be the database size and $B$ be the block size, both in bits.
If $B \ge 3\lg{n}$ and $B = O(n^c)$ for some $0<c<1$, there exists an
information theoretically secure ORAM construction such that
i) the server's space usage is
\begin{equation*}
  n \parens{10 + \Theta\parens{\frac{\log{n}}{B}}} \text{ bits;\footnotemark}
\end{equation*}
\footnotetext{The description of the original paper depends on the assumption
that $N$ is a power of two. If this assumption is not true and we pad the
database to make $N$ a power of two, the factor 10 in the server space bound
becomes 20.}
ii) the worst-case bandwidth blowup is $O(\log^2{n})$;
iii) the user's temporary space usage is $O(\log{n})$ blocks; and
iv) for any $R = \omega(\log{n})$, the probability that the user's permanent
space usage becomes larger than $R$ blocks during $\poly(n)$ logical accesses
is negligible.
\end{proposition}

\section{Succinct ORAM Construction} \label{section: succinct oram}

In this section, we prove the following theorem.

\begin{theorem} \label{main theorem}
Let $n$ be the database size and $B$ be the block size, both in bits.
If $B \ge 3\lg{n}$ and $B = O(n^c)$ for some constant $0<c<1$, then for any
$f:\N\to\R$ such that $f(n) = \omega(\log{n})$ and $f(n) = O(\log^2{n})$ and
any $g:\N\to\R$ such that $g(n) = \omega(1)$ and $g(n) =
o(\sqrt{f(n)/\log{n}})$, there exists an information theoretically secure
ORAM construction such that
i) the server's space usage is bounded by
\[
	n \parens{1 + \Theta\parens{\frac{\log{n}}{B} +
									\frac{g(n)}{\sqrt{f(n)/\log{n}}}}} \text{ bits;}
\]
ii) the worst case bandwidth blowup is $O(\log^2{n})$;
iii) the user's temporary space usage is $O(f(n))$ blocks; and
iv) for any $R = \omega(\log{n})$, the probability that the user's permanent
space usage becomes larger than $R$ blocks during $\poly(n)$ logical accesses
is negligible.
\end{theorem}

\begin{corollary}
If, in addition to the conditions of \cref{main theorem}, $B=\omega(\log{n})$,
then, the ORAM construction of \cref{main theorem} is succinct.
\end{corollary}

In the remainder of this section, $n, B, c, f(\cdot), g(\cdot)$ are as
described in \cref{main theorem}.

\subsection{Description}

For the clarity of explanation, we first describe a simplified ORAM construction
where the user needs to maintain a large amount of information locally.
Then, we obtain an ORAM construction with the claimed bounds by slightly
modifying the simplified construction.

As we mentioned in \cref{section: intro}, in a tree-based ORAM construction,
blocks on the server are stored in the nodes of a complete binary tree.
The key point of the method in this section is the choice of the tree height
$L$ and the leaf node capacity $M$.
Specifically, in the rest of this section, let
\[
  L := \ceil{\lg{\frac{N}{f(n)}}} \quad\text{and}\quad
  M := \ceil{\frac{N}{2^L} + g(n)\sqrt{\frac{NL}{2^L}}}
\]
where $N = n/B$.
We assume, for brevity, that each of $\lg{\frac{N}{f(n)}}$ and $\frac{N}{2^L} +
g(n)\sqrt{\frac{NL}{2^L}}$ is an integer.

\paragraph{Block usage.}

The ORAM is supposed to provide the user with an interface to access the
database as if it is stored in array $A$ of $B$-bit blocks (\cref{subsec:
ORAM}).
We use blocks as follows :
\begin{itemize}[noitemsep]
  \item Each block is either a \emph{data block} or a \emph{metadata block};
  \item Each data block is either a \emph{real block} or a \emph{dummy block}.
  A real block contains an entry of $A$.
  A dummy block does not contain any information on the database contents and
  is used only to hide the access pattern;
  \item Each real block is given a \emph{position label}, a value in $[2^L]$;
  \item A metadata block contains the metadata of several data blocks.
  For each data block, its metadata consists of
  \begin{description}[align=right,labelwidth=0.7cm,noitemsep]
    \item[\textmd{\textsf{type}:}] A flag indicating whether the block is real
    or dummy;
    \item[\textmd{\textsf{addr}:}] If the block is real and represents $A[i]$,
    the value of \textsf{addr} is $i$.
    If the block is a dummy, the value is arbitrary;
    \item[\textmd{\textsf{pos}:}] If the block is real with position label $i$,
    the value of \textsf{pos} is $i$.
    If the block is a dummy, the value is arbitrary.
  \end{description}
\end{itemize}

\paragraph{Data layout.}

The server maintains a tree containing data blocks, which we call \emph{data
tree}, and another tree containing metadata blocks, which we call
\emph{metadata tree}.
The data tree is used in such a way that at each point of time, it contains
most real blocks with high probability.
The user maintains \emph{stash}, which contains the real blocks that are not in
the data tree, and \emph{position table}, which contains the position labels of
all real blocks.
Below, we explain each of them more in detail.

The data tree is a complete binary tree with $2^L$ leaves.
Each node of the tree is a \emph{bucket}, which is a container that can
accommodate a certain number of blocks.
We call the buckets corresponding to the internal nodes as \emph{internal
buckets} and the buckets corresponding to the leaf nodes as \emph{leaf
buckets}.
The size of each internal bucket is $Z$ (blocks) while the size of each leaf
bucket is $M$ (blocks).
We will determine $Z$ to be 3 in \cref{subsec: succ oram user space} but for
now, we consider it as an arbitrary constant.
The data tree is represented as the bitstring derived by concatenating all
buckets in breadth first order.
As is well-known, with this representation, given an index of a node, the index
of the parent or left/right child can be derived by simple arithmetic.
The total space usage of the data tree is equal to the sum of the bucket sizes.

The metadata tree is also a complete binary tree with $2^L$ leaves.
Each node of the tree is the metadata of the data blocks in the corresponding
bucket of the data tree.
The metadata tree is represented similarly to the data tree but there is a
subtlety.
If the metadata of the blocks in a bucket has a size smaller than $B$, it is
wasteful to allocate one full block for them.
To avoid this waste, we represent metadata tree as the bitstring derived by
concatenating the metadata of all data blocks in the data tree in breadth first
order.
The space usage of the metadata tree is equal to the sum of all metadata of all
data blocks.

Each real block in the stash is maintained with its \textsf{addr} and
\textsf{pos}.
The stash can be any linear-space data structure that efficiently supports
insertion, deletion and range query by \textsf{pos}, e.g., a self balancing
binary search tree.

The position table stores the position label of the real block storing $A[i]$
in the $i$-th entry.

\paragraph{Access procedure.}

Access requests are processed in such a way that the following invariant
conditions are always satisfied:
\begin{itemize}[noitemsep]
  \item Each real block is stored either in the data tree or in the stash;
  \item If a real block with position label $\ell$ is stored in the data tree,
  it is in the bucket on the path from the root to the $\ell$-th leaf.
\end{itemize}

\begin{table}
  \centering
  \caption{The notations for access procedure}
  \label{table: notations}
  \begin{tabular}{r | l}
    $\mathrm{Pos}$
      & the position table \\
    $P(\ell)$
      & the path from the root to the $\ell$-th leaf of the data tree \\
    $P(\ell,i)$
      & the depth $i$ bucket on $P(\ell)$ (the root is at depth 0) \\
    $P(\ell,i,j)$
      & the $j$-th block in $P(\ell,i)$ (counted from one) \\
    $\mathrm{meta}[P(\ell,i)]$
      & the metadata of the blocks in $P(\ell,i)$ \\
    $|P(\ell,i)|$
      & the number of blocks in $P(\ell,i)$ ($|P(\ell,i)|=Z$ for $i < L$ and $|P(\ell,L)|=M$) \\
    $\mathrm{md}[i]$
      & the $i$-th metadata in $\mathrm{md}$ (if $\mathrm{md} = md[P(\ell,i)]$, the metadata of $P(\ell,i,j)$) \\
    $\textsc{Random}(b)$
      & returns a uniformly random $b$-bit integer \\
    $\textsc{BitReversal}(\ell)$
      & returns the $L$-bit integer derived by reversing the bits of $L$-bit integer $\ell$ \\
    $G$
      & a persistent/global variable storing the number of \textsc{Access} called so far
  \end{tabular}
\end{table}

The main routine of the access procedure is described in \cref{algorithm: main}
and the subroutines for \textsc{Access} are described in \cref{algorithm:
subroutine}.
The notations used in the access procedure are summarized in \cref{table:
notations}.\footnotemark
\footnotetext{We note that the pseudocode and notations borrow much from
existing work~\cite{Stefanov13, Ren15}.}
We use $\cdot$ to denote an arbitrary value.
For example, the metadata $(\mathrm{dummy},\cdot,\cdot)$ means any metadata
with $\mathsf{type} = \mathrm{dummy}$ (\textsf{addr} and \textsf{pos} are
arbitrary).
Though the encryption/decryption are omitted from the pseudocode for brevity,
everything on the server needs to be encrypted.
For example, in the step ``$\mathrm{md} \leftarrow \mathrm{meta}[P(\ell,i)]$'',
the user retrieves the ciphertext of $\mathrm{meta}[P(\ell,i)]$, decrypts it
and save it in the variable $\mathrm{md}$.
For brevity, we assume that every block is already initialized, i.e., each real
block is assigned a valid value with the metadata stored in the corresponding
node of the metadata tree and the position table contains the correct position
labels.

Let $b_a$ be the accessed block.
We first read the position label $\ell$ of $b_a$ from the position table and
update the position table entry to a number chosen uniformly at random from
$[L]$ (line 2--4), which will become the new position label of $b_a$ after the
access operation is finished.
By the invariant conditions above, $b_a$ is either in the stash or $P(\ell)$.
We scan $P(\ell)$ and retrieve $b_a$ if it is in $P(\ell)$ (\textsc{ReadPath}
operation in line 5).
If $b_a$ was not in $P(\ell)$, we retrieve it from the stash (line 6--9).
If the current request is a write request, we update the block contents to the
new value (line 11--12).
Then, we insert $b_a$ with the updated position label and the possibly updated
value into the stash (line 13).
After that, we perform \textsc{EvictPath} operation (line 14).
The purpose of this operation is a) to move back the blocks in the stash into
the tree and b) to move the real blocks in the tree downwards (far from the
root).
To do this, \textsc{EvictPath} retrieves all real blocks in the path
$P(\textsc{BitReversal(G)})$ (to be explained shortly) into the stash and then,
going up $P(\textsc{BitReversal(G)})$ from leaf to the root, tries to move as
many blocks in the stash into the buckets on the path.
If some blocks are left in the stash after \textsc{EvictPath}, the user keeps
them charging the permanent space usage.
Lastly, the value stored at $b_a$ is returned (line 15).

The function $\textsc{BitReversal}(\cdot)$ takes an $L$-bit integer $x$ and
returns the bit reversed version of $x$ while $G$ is the number of
\textsc{Access} operations called so far (modulo $2^L$).
Thus, if $L=8$ for example, $G$ cycles as $0,1,2,3,4,5,6,7,0,1,2,\dots$ as
$\textsc{Access}$ is called successively.
Then, $\textsc{BitReversal}(G)$ cycles as $0,4,2,6,1,5,3,7,0,4,\dots$.
The advantage of this \textsc{EvictPath} scheduling is that the eviction paths
(paths on which \textsc{EvictPath} is called) are distributed evenly, that is,
each of the $2^i$ nodes at depth $i$ is on the eviction path every $2^i$
$\textsc{Access}$ operations.
This \textsc{BitReversal}-based scheduling was first proposed by Gentry et
al.~\cite{Gentry13} and is advantageous to keep the stash size small (used
implicitly in Lemma~\ref{Ring ORAM Lemma 3}).
It also enables to simplify stash size analysis.
For security, the important thing is that $G$ (and $\textsc{BitReversal}(G)$)
is independent of the accessed database locations.

\begin{algorithm}
  \caption{Main routine}
  \label{algorithm: main}
  \begin{algorithmic}[1]
    \Function{Access}{$a,\mathrm{op},v'$}
      \State $\ell' \leftarrow \textsc{Random}(L)$
      \State $\ell \leftarrow \mathrm{Pos}[a]$
      \State $\mathrm{Pos}[a] \leftarrow \ell'$
      \Statex
      \State $v \leftarrow \textsc{ReadPath}(\ell,a)$
      \If {$v = \bot$ }
        \State find $(a,\ell,v'') \in \text{stash}$
        \Comment{there exists $(a,\ell,v'') \in \text{stash}$}
        \State $v \leftarrow v''$
        \State $\text{stash} \leftarrow \text{stash} \setminus (a,\ell,v'')$
      \EndIf
      \State $ret \leftarrow v$
      \If {$\mathrm{op} = \mathrm{write}$}
        \State $v \leftarrow v'$
      \EndIf
        \State $\text{stash} \leftarrow \text{stash} \cup (a,\ell',v)$
      \Statex
      \State \textsc{EvictPath}()
      \Statex
      \State return $ret$
    \EndFunction
  \end{algorithmic}
\end{algorithm}

\begin{algorithm}
  \caption{Subroutines for \textsc{Access}}
  \label{algorithm: subroutine}
  \begin{algorithmic}[1]
    \Function{ReadPath}{$\ell, a$}
      \State $v \leftarrow \bot$
      \For {$i \leftarrow 0$ to $L$}
        \State $\mathrm{md} \leftarrow \mathrm{meta}[P(\ell,i)]$
        \For {$j \leftarrow 1$ to $|P(\ell, i)|$}
          \State $v' \leftarrow P(\ell, i, j)$
          \If {$\mathrm{md}[j] = (\mathrm{real},a,\ell)$}
            \State $v \leftarrow v'$
            \State $\mathrm{md}[j] \leftarrow (\mathrm{dummy},\cdot,\cdot)$
          \EndIf
        \EndFor
      \State $\mathrm{meta}[P(\ell,i)] \leftarrow \mathrm{md}$
    \EndFor
    \State return $v$
  \EndFunction
  \end{algorithmic}

  \hfill

  \begin{algorithmic}[1]
    \Function{EvictPath}{ }
      \State $\ell \leftarrow G \mod{2^L}$
      \Comment{$G$ is global/persistent, and initially zero}
      \State $G \leftarrow G+1$
      \State $\ell' \leftarrow \textsc{BitReversal}(\ell)$
      \For {$i \leftarrow 0$ to $L$}
        \State $\text{stash} \leftarrow
                \text{stash} \cup \textsc{ReadBucket}(P(\ell',i))$
      \EndFor
      \For {$i \leftarrow L$ to $0$}
        \State $\textsc{WriteBucket}(P(\ell',i),\text{stash})$
      \EndFor
    \EndFunction
  \end{algorithmic}

  \hfill

  \begin{algorithmic}[1]
    \Function{ReadBucket}{$P(\ell,i)$}
      \State $S \leftarrow \emptyset$
      \State $\mathrm{md} \leftarrow \mathrm{meta}[P(\ell,i)]$
      \For {$j \leftarrow 1$ to $|P(\ell,i)|$}
        \State $v \leftarrow P(\ell,i,j)$
        \If {$\mathrm{md}[j]= (\mathrm{real},a,\ell')$ for some $a$ and $\ell'$}
          \State $S \leftarrow S \cup (a,\ell',v)$
          \State $\mathrm{md}[j] \leftarrow (\mathrm{dummy},\cdot,\cdot)$
        \EndIf
      \EndFor
      \State $\mathrm{meta}[P(\ell,i)] \leftarrow \mathrm{md}$
      \State return $S$
    \EndFunction
  \end{algorithmic}

  \hfill

  \begin{algorithmic}[1]
    \Function{WriteBucket}{$P(\ell,i)$,stash}
      \State $S \leftarrow$ blocks in the stash whose labels have the same length $i$ prefix as $\ell$
      \For{$j \leftarrow 1$ to $|P(\ell,i)|$}
        \If {$S \neq \emptyset$}
          \State pick arbitrary $(a,\ell,v) \in S$
          \State $P(\ell,i,j) \leftarrow v$
          \State $\mathrm{md}[j] \leftarrow (\mathrm{real},a,\ell)$
          \State $S \leftarrow S \setminus (a,\ell,v)$
        \Else
          \State $P(\ell,i,j) \leftarrow \text{garbage}$
          \State $\mathrm{md}[j] \leftarrow (\mathrm{dummy},\cdot,\cdot)$
        \EndIf
      \EndFor
      \State $md[P(\ell,i)] \leftarrow \mathrm{md}$
    \EndFunction
  \end{algorithmic}
\end{algorithm}

\paragraph{Outsourcing position table.}

In the construction described so far, the user space usage is much larger than
the bound claimed in \cref{main theorem} since the user needs to maintain the
position table locally.
To obtain \cref{main theorem}, we modify the construction so that the position
table is stored on the server using the sub-ORAM in \cref{prop: Path ORAM},
e.g., the Path ORAM~\cite{Stefanov13}.
Access procedure is the same except that the line 3--4 of \textsc{Access}
($\ell \leftarrow Pos[a]$ and $Pos[a] \leftarrow \ell)$ is replaced by a
sub-ORAM write access.

\subsection{Security} \label{subsec: succ oram security}

Fix $t>0$.
Let $\mathbf{a}$ be a length $t>0$ sequence of logical addresses to be accessed
and $\mathbf{a}'$ be the corresponding sequence of physical addresses (indices
of the server memory) to be accessed.
The sequence $\mathbf{a}'$ is determined by $\mathbf{a}$ and the randomness
used by the ORAM simulator.
To prove the information theoretic security, it suffices to show that
$\mathbf{a}'$ really does not depend on $\mathbf{a}$.
The sequence $\mathbf{a}'$ consists of $\mathbf{a}'_1$, the physical addresses
accessed in step 3--4 of \textsc{Access} and $\mathbf{a}'_2$, those accessed in
the rest parts of \textsc{Access}.
The addresses $\mathbf{a}'_1$ is determined by the sub-ORAM access procedure
and is independent of $\mathbf{a}$ due to the information theoretic security of
the sub-ORAM.
The addresses $\mathbf{a}'_2$ consists of addresses accessed by
$\textsc{ReadPath}(\ell,a)$ and $\textsc{EvictPath}()$.
$\textsc{ReadPath}(\ell,a)$ accesses the path $P(\ell)$, which is determined by
$\ell$, the position label of the accessed block.
Since the position labels are chosen independently and uniformly at random, the
$\textsc{ReadPath}$ accesses are independent of $\mathbf{a}$.
$\textsc{EvictPath}$ accesses $P(\textsc{BitReversal}(G))$, which is determined
by $G$, the number of times $\textsc{Access}$ was called (modulo $2^L$).
Thus, the accesses of $\textsc{EvictPath}$ is also independent of $\mathbf{a}$.
Therefore, $\mathbf{a}'$ is independent of $\mathbf{a}$.

\subsection{Server Space}

First, it is helpful to observe the followings:
\begin{equation}
  \log{N} = \Theta(\log{n}), \quad
  L = \Theta(\log{n}), \quad
  M = \Theta(f(n)) \label{useful bounds}.
\end{equation}
Remember that the server holds the data tree, the metadata tree and the
position table.

The total size of the internal (resp. leaf) buckets is $Z(2^L-1)$ (resp.
$M2^L$) blocks.
Since
\begin{align*}
  & Z(2^L-1) < Z2^L = ZN/f(n),
  & M2^L = N+g(n)\sqrt{NL2^L} = N\parens{1+\Theta\parens{\frac{g(n)}{\sqrt{f(n)/\log{n}}}}},
\end{align*}
the number of the blocks in the data tree is bounded by
\begin{equation*}
    ZN/f(n) + N\parens{1+\Theta\parens{\frac{g(n)}{\sqrt{f(n)/\log{n}}}}}
  = N\parens{1 + \Theta\parens{\frac{1}{f(n)} + \frac{g(n)}{\sqrt{f(n)/\log{n}}}}}.
\end{equation*}

The metadata for each data block takes 1 bit for \textsf{type}, $\ceil{\lg{N}}$
bits for \textsf{addr} and $L$ bits for \textsf{pos}.
The total is $\Theta(\log{n})$ bits, which is $\Theta(\frac{\log{n}}{B})$
blocks.
Thus, the number of bits in the data tree and the metadata tree combined is
\begin{equation*}
    BN\parens{1 + \Theta\parens{\frac{1}{f(n)} + \frac{g(n)}{\sqrt{f(n)/\log{n}}}}} \parens{1 + \Theta\parens{\frac{\log{n}}{B}}}
  =  n\parens{1 + \Theta\parens{\frac{\log{n}}{B} + \frac{g(n)}{\sqrt{f(n)/\log{n}}}}}.
\end{equation*}

The position labels take $NL = n\frac{L}{B} \le n\frac{\lg{n}}{B}$ bits.
By \cref{prop: Path ORAM}, the sub-ORAM containing the position table takes
$\Theta(n\frac{\log{n}}{B})$ bits.
Thus, the server space is $n(1 + \Theta(\frac{\log{n}}{B} +
\frac{g(n)}{\sqrt{f(n)/\log{n}}}))$ bits.

\subsection{Bandwidth Blowup}

The bandwidth cost of each of \textsc{ReadPath} and \textsc{EvictPath} is
proportional to the sum of the numbers of the blocks in a root--leaf path in
the data tree and the metadata tree.
The number for the data tree is $ZL+M = O(\log{n}) + O(f(n)) = O(f(n))$.
The number for the metadata tree is around $\frac{2\lg{N}+1}{B} = o(1)$ factor
of that for the data tree.
The bandwidth cost for accessing the position table is $O(\log^2{n})$ by
\cref{prop: Path ORAM}.
Therefore, the bandwidth blowup of \textsc{Access} is $O(\log^2{n})$.

\subsection{User Space} \label{subsec: succ oram user space}

The temporary user space usage is proportional to the sum of the numbers of the
blocks in a root--leaf path in the data tree and the metadata tree.
As is shown in the bandwidth analysis, the latter is bounded by $O(f(n))$.

In the rest of this subsection, we bound the permanent user space usage, i.e.,
the stash size.
First, we import some concepts and tools from~\cite{Stefanov13} and
\cite{Ren15}.
Fix a sequence of input logical access requests.
Later, we will specify a concrete request sequence that we use for the
analysis.
Let $\mathrm{ORAM}_Z$ be the real ORAM construction that we are analyzing and
$\mathrm{ORAM}_\infty$ be the hypothetical ORAM construction derived by
modifying the size of each bucket in $\mathrm{ORAM}_Z$ to $\infty$.
Let $S_Z$ (resp. $S_\infty$) be the $\mathrm{ORAM}_Z$ (resp.
$\mathrm{ORAM}_\infty$) after processing the access requests.
Note that $\mathrm{ORAM}_Z$ (resp. $\mathrm{ORAM}_\infty$) is an ORAM
construction while $S_Z$ (resp. $S_\infty$) is the state of the construction at
a particular point of time.
We write $G$ to denote a post-processing algorithm that takes $S_Z$ and
$S_\infty$ and modify $S_\infty$ in the following way.
The algorithm $G$ enumerates the buckets in $\mathrm{ORAM}_\infty$ in reverse
breadth first order.
We define $b^Z_1$ to be the root bucket of $\mathrm{ORAM}_Z$ and $b^Z_{2i}$
(resp. $b^Z_{2i+1}$) to be the left (resp. right) child of $b^Z_i$.
We define $b^Z_0$ to be the stash of $\mathrm{ORAM}_Z$.
For $i\in[2^{L+1}]$, $b^\infty_i$ is defined similarly.
For each $i$ from $2^{L+1}-1$ to 1, $G$ processes each block $v\in b^\infty_i$
as follows: i) if $v \in b^Z_i$, $v$ is left as it is; ii) if in $S_Z$, $v$ is
stored in some proper ancestor of $b^Z_i$, $G$ moves $v$ to
$b^\infty_{\floor{i/2}}$, i.e., $b^Z_i$'s parent.  If the number of blocks left
in $b^\infty_i$ after such transportation is less than $Z$, $G$ outputs an
error; iii) if, in $S_Z$, $v$ is not stored in any ancestor of $b^Z_i$, $G$
outputs an error.
We denote the output of $G$ with input $S_Z$ and $S_\infty$ as
$G_{S_Z}(S_\infty)$.
For $S\in\{S_Z,G_{S_Z}(S_\infty)\}$, we define $st(S)$ to be the number of
blocks in the stash of $S$.
Due to the following lemma, $st(S_Z)$ and $st(G_{S_Z}(S_\infty))$ are
equivalent as random variables.

\begin{lemma}[\cite{Ren15} Lemma 1]
If the randomness used in $\mathrm{ORAM}_Z$ is the same as
$\mathrm{ORAM}_\infty$, i.e., the position labels assigned to the accessed
blocks are the same in the two ORAM constructions, then $G$ does not output an
error and $S_Z = G_{S_Z}(S_\infty)$.
\end{lemma}

\noindent Let a subtree be a connected subgraph of the complete binary tree
with $2^L$ leaves that contains the root.
For a subtree $T$, let $C(T)$ be the number of blocks that can be stored in the
corresponding buckets of $\mathrm{ORAM}_Z$, and $X(T)$ be the number of blocks
that are stored in the corresponding buckets of $S_\infty$.
Note that $C(T)$ is a constant while $X(T)$ is a random variable.
Also, let $n(T)$ denote the number of nodes in $T$.

\begin{lemma}[\cite{Ren15} Lemma 2]
For any integer $R>0$, $st(G_{S_Z}(S_\infty)) > R$ iff there exists a subtree
$T$ such that $X(T) > C(T)+R$.
\end{lemma}

\noindent We call those subtrees that contain only internal nodes (of the
enclosing complete binary tree) as internal subtrees.

\begin{lemma}[\cite{Ren15} Lemma 3] \label{Ring ORAM Lemma 3}
For any internal subtree $T$, $E[X(T)] \le n(T)/2$.
\end{lemma}

\noindent Let the working set of a sequence of access requests
$(\mathrm{op}_i,\mathrm{addr}_i,\mathrm{val}_i)_i$ be the set
$\{\mathrm{addr}_i\}_i$.

\begin{lemma}[\cite{Stefanov13} Lemma 3] \label{Path ORAM Lemma 3}
Among all access request sequences of working set size $t$, the probability
$\Pr[st(S_Z)>R]$ is maximized by the sequence that contains exactly one access
to each of the $t$ different addresses.
\end{lemma}

\noindent Because of \cref{Path ORAM Lemma 3}, we fix the input access request
sequence to $(\mathrm{op}_i,i,\mathrm{val}_i)_{i\in[N]}$ without loss of
generality.
($\mathrm{op}_i$ and $\mathrm{val}_i$ are arbitrary since they do not affect
the stash size.)

Now we prove
\begin{equation} \label{target bound}
  \Pr[st(S_Z)>R] = n^{-\omega(1)}
\end{equation}
for $R = \omega(\log{n})$.
Remember that \eqref{target bound} is a bound on the stash size at a particular
point of time.
Given \eqref{target bound}, the probability that the stash size becomes larger
than $R$ at \emph{any} point in $\poly(n)$ logical accesses is also bounded as
$n^{-\omega(1)}$ by union bound.

Let \textsf{G} be the event that no leaf bucket of $S_\infty$ contains more
than $M$ blocks.
Let \textsf{B} be the complement of \textsf{G}.

\begin{lemma} \label{balls-into-bins bound}
  $\Pr[\mathsf{B}] = n^{-\omega(1)}$.
\end{lemma}

\begin{proof}
Consider $\mathrm{ORAM}_\infty$ just before post-processing.
For $i \in [2^L]$, let $\mathrm{load}_i$ be the number of real blocks in the
$i$-th leaf bucket and $\mathrm{ctr}_i$ be the number of real blocks with
position label $i$.
Since a real block can be stored in the $i$-th leaf bucket only if it has
position label $i$, $\mathrm{load}_i \le \mathrm{ctr}_i$.
For $i \in [2^L]$ and $j \in [N]$ let $\mathrm{ctr}_{i,j}$ be the indicator
random variable of the event that the $j$-th accessed real block is assigned
position label $i$.
Clearly, $\mathrm{ctr}_i = \sum_j \mathrm{ctr}_{i,j}$ for each $i \in [2^L]$,
and $E[\mathrm{ctr}_{i,j}] = \Pr[\mathrm{ctr}_{i,j}=1] = 1/2^L$ for each $i \in
[2^L]$ and $j \in [N]$.
Thus, $E[\mathrm{ctr}_i] = E[\sum_j \mathrm{ctr}_{i,j}] = \sum_j
E[\mathrm{ctr}_{i,j}] = N/2^L$ for each $i \in [2^L]$.
For each $i\in[2^L]$, $\{\mathrm{ctr}_{i,j}\}_{j\in[N]}$ are mutually
independent.
The lemma follows as
\begin{align*}
        \Pr[\mathsf{B}]
  & =   \Pr[\cup_{i\in[2^L]} \mathrm{load}_i>M] \\
  & \le \Pr[\cup_{i\in[2^L]} \mathrm{ctr}_i > M] \\
  & \le \Sigma_{i\in[2^L]} \Pr[\mathrm{ctr}_i > M] \\
  & =   \sum_{i\in[2^L]} \Pr \brackets{\mathrm{ctr}_i > \frac{N}{2^L}\parens{1+g(n)\sqrt{\frac{L2^L}{N}}}} \\
  & \le 2^L \exp\parens{-(1/3)g(n)^2 L} \\
  & \le n \exp(-\omega(1)\ln{n}) \\
  & =   n^{-\omega(1)}.
\end{align*}
We used Chernoff bound in the fifth step.
\end{proof}

Let $\mathcal{T}$ be the set of all subtrees and $\mathcal{T}'$ be the set of
all internal subtrees.
Then,
\begin{align}
        \Pr[st(S_Z)>R]
  & =   \Pr[st(G_{S_Z}(S_\infty))>R] \nonumber \\
  & =   \Pr[\cup_{T\in\mathcal{T}} X(T)>C(T)+R] \nonumber \\
  & \le \Pr[\cup_{T\in\mathcal{T}} X(T)>C(T)+R | \mathsf{G}] + \Pr[\mathsf{B}] \nonumber \\
  & =   \Pr[\cup_{T\in\mathcal{T}'} X(T)>C(T)+R | \mathsf{G}] + \Pr[\mathsf{B}] \nonumber \\
  & \le \Sigma_{T\in\mathcal{T}'} \Pr[X(T)>C(T)+R | \mathsf{G}] + \Pr[\mathsf{B}] \nonumber \\
  &\le \sum_{m\ge 1} 4^m \max_{\substack{T\in\mathcal{T}' \\ n(T)=m}} \Pr[X(T)>C(T)+R|\mathsf{G}] + \Pr[\mathsf{B}]. \label{stash bound 1}
\end{align}
In the last step, we used the fact that the number of ordered binary trees with
$m$ nodes is bounded by $4^m$.

\begin{lemma} \label{stash bound 4}
For any internal subtree $T$ with $n(T) = m$,
\begin{equation*}
      \Pr[X(T)>C(T)+R|\mathsf{G}]
  \le (2Z)^{-R} \exp(-m(Z\ln{2Z} + 1/2 - Z)) \Pr[\mathsf{G}]^{-1}.
\end{equation*}
\end{lemma}

\begin{proof}
For any $t>0$,
\begin{align}
        \Pr[X(T) >C(T)+R |\mathsf{G}]
  & =   \Pr[e^{tX(T)} > e^{t(C(T)+R)}| \mathsf{G}] \nonumber \\
  & \le E[e^{tX(T)} | \mathsf{G}] e^{-t(C(T)+R)} \nonumber \\
  & \le E[e^{tX(T)}] \Pr[\mathsf{G}]^{-1} e^{-t(C(T)+R)}. \label{stash bound 2}
\end{align}
For $j\in[N]$, let $X_j(T)$ be the indicator random variable of the event that,
in $S_\infty$, the $j$-th accessed real block is in $T$ and let $p_j :=
\Pr[X_j(T)=1]$.
Clearly, $\sum_j X_j(T) = X(T)$ and $E[X(T)]=E[\sum_j X_j(T)] = \sum_j
E[X_j(T)] = \sum_j p_j$.
The random variable $X_j(T)$ depends only on $j$ and the position label of the
$j$-th accessed real block.
Thus, $\{X_j(T)\}_{j\in[N]}$ are mutually independent.
Then,
\begin{align}
        E[e^{tX(T)}]
  & =   E[e^{t\sum_{j\in[N]}X_j(T)}] \nonumber \\
  & =   E[\Pi_{j\in[N]}e^{tX_j(T)}] \nonumber \\
  & =   \Pi_{j\in[N]} E[e^{tX_j(T)}] \nonumber \\
  & =   \Pi_{j\in[N]} (p_j(e^t - 1)+1) \nonumber \\
  & \le \Pi_{j\in[N]} \exp(p_j(e^t-1)) \nonumber \\
  & =   \exp((e^t-1)\Sigma_{j\in[N]}p_i) \nonumber \\
  & =   \exp((e^t-1)E[X(T)]). \label{stash bound 3}
\end{align}
We used the independence of $\{X_j(T)\}_{j\in[N]}$ in the third step.
Let $m := n(T)$.
From bounds \eqref{stash bound 2}, \eqref{stash bound 3} and \cref{Ring ORAM
Lemma 3}, $\Pr[X(T)>C(T)+R|\mathsf{G}]$ is bounded by
\begin{equation*}
    \exp((e^t-1)m/2)e^{-t(mZ+R)}\Pr[\mathsf{G}]^{-1}
  = \exp(-tR) \exp(-m(tZ-(1/2)(e^t-1))) \Pr[\mathsf{G}]^{-1}.
\end{equation*}
The lemma follows by setting $t = \ln{2Z}$.
\end{proof}

If $Z=3$, $q := Z\ln{2Z}+1/2-Z-\ln{4} = 1.4889 \dots > 0$.

By \eqref{stash bound 1} and \cref{stash bound 4},
$\Pr[st(S_Z)>R]$ is bounded by
\begin{equation*}
    \sum_{m \ge 1} 4^m 6^{-R} \exp(-m(q+\ln{4})) \Pr[\mathsf{G}]^{-1} + \Pr[\mathsf{B}]
  < \frac{(1/6)^R}{1-e^{-q}} \Pr[\mathsf{G}]^{-1} + \Pr[\mathsf{B}].
\end{equation*}
By \cref{balls-into-bins bound}, the bound above is $n^{-\omega(1)}$ if
$R=\omega(\log{n})$.

\section{Succincter ORAM Construction} \label{section: succincter oram}

In this section, we prove the following theorem.

\begin{theorem} \label{second main theorem}
Let $n$ be the database size and $B$ be the block size, both in bits.
If $B \ge 3\lg{n}$ and $B = O(n^c)$ for some $0 < c < 1$, then for any
$f:\N\to\R$ such that $f(n) = \omega(\log\log{n})$ and $f(n) = O(\log^2{n})$,
there exists an information theoretically secure ORAM construction for which
i) the server's space usage is bounded by
\[
	n \parens{1 + \Theta\parens{\frac{\log{n}}{B} +
								\frac{\log\log{n}}{f(n)}}} \text{ bits;}
\]
ii) the worst case bandwidth blowup is $O(\log^2{n})$;
iii) the user's temporary space usage is $O(\log{n}+f(n))$ blocks; and
iv) for any $R = \omega(\log{n})$, the probability that the user's permanent
space usage becomes larger than $R$ blocks during $\poly(n)$ logical accesses
is $n^{-\omega(1)}$.
\end{theorem}

\begin{corollary}
If, in addition to the conditions of \cref{second main theorem}, $B =
\omega(\log{n})$, then, the ORAM construction of \cref{second main theorem} is
succinct.
\end{corollary}

\cref{second main theorem} is stronger than \cref{main theorem}.
For example, if $B = \Omega(\log^2{n})$ and $f(n)=\Theta(\log{n}\log\log{n})$,
the server space bound of \cref{second main theorem} implies that the extra
server space is $\Theta(n/\log{n})$ and the user temporary space usage is
$\Theta(\log{n}\log\log{n})$.
In contrast, the extra server space bound of \cref{main theorem} is
$\omega(n/\sqrt{\log{n}})$ even if we allow the user's temporary space to
become $\Theta(\log^2{n})$.

In the rest of this section, $n, B, f(\cdot)$ are as described in the statement
of \cref{second main theorem}.

In the following exposition, we often refer to \cref{section: succinct oram} to
avoid repetition.
We recommend the readers to read \cref{section: succinct oram} beforehand.

\subsection{Description}

As in Section~\ref{section: succinct oram}, we first explain a simplified
version with a large user space usage, and construct the full version that
achieves the claimed bounds from the simplified version.

Let
\[
  L := \ceil{\lg(N/f(n))} \quad\text{and}\quad
  M := \ceil{N/2^L + (1+\eps)\lg{L}}
\]
where $N = n/B$ and $\eps>0$ is a constant.
We assume, for brevity, that each of $\lg(N/f(n))$ and $N/2^L + (1+\eps)\lg{L}$
is an integer.

\paragraph{Block usage.}

The block usage is the same as the ORAM construction described in
\cref{section: succinct oram} except that each real block is given \emph{two}
position labels instead of one.
We call them the \emph{primary position label} and the \emph{secondary position
label}.
Only the primary position labels are stored in the metadata blocks.

\paragraph{Data layout.}

The data layout is basically the same in \cref{section: succinct oram}.
We only explain the difference from \cref{section: succinct oram}.

First, the position table stores both the primary position labels and the
secondary position labels.

Second, the user maintains an additional table called \emph{counter table}.
It is a size $2^L$ array whose $i$-th entry is the number of real blocks with
primary position label $i$.

Last, since the value of each of $L$ and $M$ is different from that in
\cref{section: succinct oram}, the tree/bucket size is changed accordingly.

\paragraph{Access procedure.}

The same invariant conditions as \cref{section: succinct oram} are maintained
except that the ``position label'' in the second condition is replaced by
``primary position label''.

The main routine is described in \cref{algorithm: main (two choices)}.
The array $\mathrm{Pos}$ and the subroutines \textsc{Random}, \textsc{ReadPath}
and \textsc{EvictPath} are the same as in \cref{section: succinct oram} while
$\mathrm{Ctr}$ is the counter table.
We let $P(\ell)$ denote the path from the root to the $\ell$-th leaf in the
data tree.
For brevity, we assume that every block is already initialized, i.e., each real
block is assigned a valid value with the metadata stored in the corresponding
node of the metadata tree and the position table and counter table contain the
correct values.

Let $b_a$ be the accessed block.
We first retrieve the two position labels $\ell_1$ and $\ell_2$ of $b_a$ from
the position table and update each of the two position table values to a number
chosen independently and uniformly at random from $[L]$, which will become the
new position labels of $b_a$ (line 2--4).
One of $\ell_1$ and $\ell_2$ is the primary position label and the other is the
secondary position label but we do not know (and do not need to know) which is
which.
By the invariant conditions, $b_a$ is either in the stash or in $P(\ell_1)$ or
$P(\ell_2)$.
We scan $P(\ell_1)$ and $P(\ell_2)$ and retrieve $b_a$ from $P(\ell_i)$ if the
primary position label is $\ell_i$ and $b_a$ is in $P(\ell_i)$ (line 5).
If $b_a$ is not found in the paths, it must be in the stash and we retrieve it
from the stash (line 11--13).
At this point, we know the primary position label $\ell$ of $b_a$ (since it is
written in the \textsf{pos} entry of the block) and we decrement the $\ell$-th
entry of the counter table, determine the new primary position label $\ell'_i$
and increment the $\ell'_i$-th entry of the counter table (line 14--17).
After, that, we update the block contents if it is a write request (line
19--20), insert $b_a$ into the stash (line 21), call \textsf{EvictPath} (line
22) and returns the retrieved block content (line 23) all in the same way as
Algorithm~\ref{algorithm: main}.

\paragraph{Outsourcing the position/counter table.}

In the full version of the construction, the position table and the counter
table are stored on the server using the sub-ORAM in \cref{prop: Path ORAM}.
Every access to each of these tables is done using the sub-ORAM access
procedure.

\begin{algorithm}
  \caption{Main routine (two choices)}
  \label{algorithm: main (two choices)}
  \begin{algorithmic}[1]
    \Function{Access}{$a,\mathrm{op},v'$}
      \State $\ell'_1 \leftarrow \textsc{Random}(L), \ell'_2 \leftarrow \textsc{Random}(L)$
      \State $(\ell_1,\ell_2) \leftarrow \mathrm{Pos}[a]$
      \State $\mathrm{Pos}[a] \leftarrow (\ell'_1,\ell'_2)$
      \Statex
      \State $v_1 \leftarrow \textsc{ReadPath}(\ell_1,a), v_2 \leftarrow \textsc{ReadPath}(\ell_2,a)$
      \If {$v_1 \neq \bot$}
        \State $(v, \ell) \leftarrow (v_1, \ell_1)$
      \ElsIf {$v_2 \neq \bot$}
        \State $(v, \ell) \leftarrow (v_2, \ell_2)$
      \Else
        \State Find $(a,\ell'',v'') \in \text{stash}$
        \Comment{There exists $(a,\ell'',v'') \in \text{stash}$}
        \State $(v, \ell) \leftarrow (v'', \ell'')$
        \State $\text{stash} \leftarrow \text{stash} \setminus (a,\ell'',v'')$
      \EndIf
      \Statex
      \State $\mathrm{Ctr}[\ell] \leftarrow \mathrm{Ctr}[\ell]-1$
      \State $c_1 \leftarrow \mathrm{Ctr}[\ell'_1], c_2 \leftarrow \mathrm{Ctr}[\ell'_2]$
      \State $i \leftarrow \argmin \{c_1,c_2\}$
      \State $\mathrm{Ctr}[\ell'_i] \leftarrow c_i + 1$
      \Statex
      \State $ret \leftarrow v$
      \If {$\mathrm{op} = \mathrm{write}$}
        \State $v \leftarrow v'$
      \EndIf
        \State $\text{stash} \leftarrow \text{stash} \cup (a,\ell'_i,v)$
      \Statex
      \State \textsc{EvictPath}()
      \Statex
      \State return $ret$
    \EndFunction
  \end{algorithmic}
\end{algorithm}

\subsection{Security}

The security proof of the current ORAM construction is almost the same as in
\cref{subsec: succ oram security}.
The only difference in the situation is that now, the sequence of accessed
addresses $\mathbf{a}'_2$ depends on two position labels instead of one.
Anyway, these position labels are distributed independently and uniformly at
random and thus, are independent of $\mathbf{a}$.

\subsection{Server Space}

The bounds \eqref{useful bounds} still hold.

The number of blocks in the leaf buckets is
\begin{align*}
  M2^L & = N\parens{1+(1+\eps)\frac{\lg{L}}{f(n)}} \nonumber \\
       & = N\parens{1+ \Theta\parens{\frac{\log\log{n}}{f(n)}}}.
\end{align*}
The number of blocks in the internal buckets is $Z(2^L-1) < ZN/f(n)$, which is
$O(\frac{\log\log{n}}{f(n)})$.
Thus, the data tree size is bounded by $N(1+ \Theta(\frac{\log\log{n}}{f(n)}))$
blocks.
As in \cref{section: succinct oram}, the metadata size of each data block is
$\Theta(\frac{\log{n}}{B})$ blocks.
Thus, the number of blocks in the data tree and the metadata tree combined is
at most $1+\Theta(\frac{\log{n}}{B})$ times larger than $N(1+
\Theta(\frac{\log\log{n}}{f(n)}))$, which is
\begin{equation*}
  n \parens{1+ \Theta\parens{\frac{\log{n}}{B} + \frac{\log\log{n}}{f(n)}}} \text{ bits}.
\end{equation*}

Position labels take $2NL = 2nL/B \le 2n\frac{\log{n}}{B}$ bits while counter
table values take $2^L\ceil{\lg{N}} = N\ceil{\lg{N}}/f(n) \le N = n/B$ bits.
By \cref{prop: Path ORAM}, the sub-ORAM containing the position table (resp.
counter table) takes $\Theta(n\frac{\log{n}}{B})$ (resp. $\Theta(n/B)$) bits.

Therefore, the server space usage is bounded by $n (1+ \Theta(\frac{\log{n}}{B}
+ \frac{\log\log{n}}{f(n)}))$ bits.

\subsection{Bandwidth Blowup}

By the same argument as in the bandwidth analysis, the bandwidth cost of each
of \textsc{ReadPath} and \textsc{EvictPath} is proportional to $ZL+M =
O(\log{n}+f(n))$ (in blocks).
By \cref{prop: Path ORAM}, the bandwidth cost of access to each of the position
table and the counter table is $O(\log^2{n})$.
Thus, the bandwidth blowup is $O(\log^2{n})$.

\subsection{User Space}

By the same argument as in \cref{subsec: succ oram user space}, the temporary
user space is proportional to $ZL+M = O(\log{n}+f(n))$.

In the rest of the subsection, we bound the permanent user space, i.e., the
stash size.
Using the current ORAM construction, define $\mathrm{ORAM}_Z$,
$\mathrm{ORAM}_\infty$, $S_Z$ and $S_\infty$ analogously to \cref{subsec: succ
oram user space}.
Then, we prove $\Pr[st(S_Z)>R] = n^{-\omega(1)}$ for $R=\omega(\log{n})$.
Most arguments in \cref{subsec: succ oram user space} can be reused and we
focus on the differences.

First, \cref{balls-into-bins bound} still holds for the current
construction but the proof is different from \cref{subsec: succ oram user
space}.

\begin{proof}[Proof of \cref{balls-into-bins bound} for the two-choice
construction]
Define $\mathrm{load}_i$, $\mathrm{ctr}_i$ and $\mathrm{ctr}_{i,j}$ in the same
way as we did in the proof of \cref{balls-into-bins bound} except that the
primary position labels are used instead of the position labels.
By the same argument as the proof of \cref{balls-into-bins bound}, it
suffices to prove $\Pr[\cup_{i\in[2^L]} \mathrm{ctr}_i > M] = n^{-\omega(1)}$.

We apply an existing bound for the heavily loaded case of the
\emph{balls-into-bins game with two choices}.
In the balls-into-bins game with $m$ balls and $n$ bins (with one choice), each
of the $m$ balls is thrown into one of the $n$ bins chosen uniformly and
independently at random.
In the balls-into-bins game with two choices, for each ball, two bins are
chosen uniformly and independently at random.
Then, the ball is thrown into the least loaded bin.
The \emph{gap} of a balls-into-bins game with $m$ balls and $n$ bins is defined
to be the difference between the number of balls in the bin with the maximum
load and the average number of balls in a bin, i.e., $m/n$.
Berenbrink et al.~\cite{Berenbrink00} proved the following proposition.

\begin{proposition}
In the two-choice balls-into-bins game with $m$ balls and $n$ bins, for any
$c$, $\Pr[\text{gap} > \lg\lg{n} + \gamma(c)] < 1/n^c$, where $\gamma(c)$ is a
constant that depends only on $c$.
\end{proposition}

\begin{corollary} \label{balls-into-bins bound (two-choice)}
In the two-choice balls-into-bins game with $m$ balls and $n$ bins,
$\Pr[\text{gap} > (1+\eps)\lg\lg{n}] = n^{-\omega(1)}$ for any $\eps>0$.
\end{corollary}

\noindent After processing the access requests, the $2^L$ values in the counter
table are distributed in exactly the same way as the bin loads after the
balls-into-bins game with two choices with $N$ balls and $2^L$ bins.
(Remember that each of the $N$ logical addresses is accessed exactly once.)
Thus,
\begin{align*}
        \Pr[\mathsf{B}]
  & \le \Pr\brackets{\cup_{i\in[2^L]} \mathrm{ctr}_i > M} \\
  & =   \Pr\brackets{\cup_{i\in[2^L]} \mathrm{ctr}_i > N/2^L+(1+\eps)\lg{L}} \\
  & =   (2^L)^{-\omega(1)} \\
  & =   n^{-\omega(1)}
\end{align*}
where we used \cref{balls-into-bins bound (two-choice)} in the third step.
\end{proof}

Next, we modify \cref{stash bound 4}.\footnotemark
\footnotetext{We need to do this since $\{X_j(T)\}_{j\in[N]}$ (defined
analogously in \cref{stash bound 4}) is not mutually independent due to
the two-choice strategy.}

\begin{lemma} \label{stash bound 4 (two-choice)}
For any internal subtree $T$ with $n(T)=m$,
  \begin{equation*}
        \Pr[X(T)>C(T)+R|\mathsf{G}]
    \le (Z)^{-R} \exp(-m(Z\ln{Z} + 1 - Z)) \Pr[\mathsf{G}]^{-1}.
  \end{equation*}
\end{lemma}

\begin{proof}
The part where the arguments in \cref{subsec: succ oram user space} breaks down
is \eqref{stash bound 3}.
In \cref{subsec: succ oram user space}, we used the mutual independence of
$\{X_j\}_{j\in[N]}$ for the third step of \eqref{stash bound 3} but here,
$\{X_j\}_{j\in[N]}$ are not mutually independent.
We fix this problem as follows.

Let $\mathrm{ORAM}'_\infty$ be another hypothetical ORAM construction derived
by modifying $\mathrm{ORAM}_\infty$ so that every time a just accessed real
block $b$ is inserted into the stash, another block called $b$'s \emph{shadow}
is also inserted into the stash.
If $b$ is given the primary position label $\ell_1$ and the secondary position
label $\ell_2$ at the time $b$'s shadow $b'$ is inserted into the stash, $b'$
is given the primary position label $\ell_2$ and the secondary position label
$\ell_1$.
A shadow is evicted in the same way as a real block but it does not affect the
counter table.
Let $S'_\infty$ be $\mathrm{ORAM}'_\infty$ after processing the access
requests.
Since each of the $N$ real blocks is accessed exactly once, each real block in
$S'_\infty$ has one shadow.
For $j\in[N]$, let $Y_j(T)$ (resp. $Y'_j(T)$) be the indicator random variable
of the event that, in $S'_\infty$, the $j$-th accessed real block (resp. the
$j$-th accessed real block's shadow) is in subtree $T$.
Since shadows do not affect real blocks' move, $X_j$ and $Y_j$ are equivalent
random variables.
Also, since the primary and the secondary position label of each accessed block
is chosen independently and uniformly at random, $Y_j+Y'_j$ is distributed
equally as $U_j+U'_j$ where $U_j$ and $U'_j$ are independent random variables,
each distributed equally with the $X_j$ in the proof of \cref{stash bound 4}.
Thus, with $Y(T):=\sum_jY_j(T)$,
\begin{align}
      E[e^{tX(T)}] & \le E[e^{t\sum_j(Y_j(T)+Y'_j(T))}] \nonumber \\
  & = E[e^{t\sum_j(U_j+U'_j)}] \nonumber \\
  & = E[e^{t\sum_j U_j}] E[e^{t\sum_j U'_j}] \nonumber \\
  & = E[e^{t\sum_j U_j}]^2 \nonumber \\
  & = \exp((e^t-1)2E[X(T)]) \label{stash bound 3 (two-choices)}
\end{align}
where we used~\eqref{stash bound 3} in the last step.

Then, from \eqref{stash bound 2}, \eqref{stash bound 3 (two-choices)} and
\cref{Ring ORAM Lemma 3}, $\Pr[X(T)>C(T)+R|\mathsf{G}]$ is bounded by
\begin{equation*}
    \exp((e^t-1)m-t(mZ+R))\Pr[\mathsf{G}]^{-1}
  = \exp(-tR) \exp(-m(tZ-(e^t-1))) \Pr[\mathsf{G}]^{-1}.
\end{equation*}
The lemma follows by setting $t = \ln{Z}$.
\end{proof}

If $Z=4$, $q = Z\ln{Z} + 1 - Z - \ln{4} > 1.15888 \dots > 0$.
By \eqref{stash bound 1} and \cref{stash bound 4 (two-choice)},
$\Pr[st(S_Z)>R]$ is bounded by
\begin{equation*}
    \sum_{m \ge 1} 4^m 4^{-R} \exp(-m(q+\ln{4})) \Pr[\mathsf{G}]^{-1} + \Pr[\mathsf{B}]
  < \frac{(1/4)^R}{1-e^{-q}} \Pr[\mathsf{G}]^{-1} + \Pr[\mathsf{B}].
\end{equation*}
By \cref{balls-into-bins bound}, the bound above is $n^{-\omega(1)}$ if
$R=\omega(\log{n})$.

\section{Practicality of the Proposed Methods} \label{section: non-asymptotic}

\cref{table: performance with concrete params} shows the performance of the
proposed methods, the Path ORAM~\cite{Stefanov13} and the Ring
ORAM~\cite{Ren15} with concrete parameters.
The Ring ORAM has asymptotically the same performance as the Path ORAM but it
achieves constant factor smaller bandwidth at the cost of larger server space.
It is easy to integrate the main technique of the Ring ORAM to the internal
nodes of the proposed methods\footnotemark and we also show the performance of
these variants.
\footnotetext{Specifically, we modify the access procedure to access only one
block per each bucket (instead of all blocks in the bucket) by permuting the
blocks in each bucket. For this technique to work, we need to introduce, for
each bucket, additional space dedicated only for dummy blocks and this is why
we cannot apply this technique to the leaves maintaining succinctness.}

The table contains ``rigorous'' and ``aggressive'' parameter settings.
Rigorous parameters were derived from theoretical analysis with additional care
for constant factors.
The aggressive parameters for existing methods were taken from the experiments
in the original papers.
We chose the aggressive parameters for the proposed methods by simulation: we
simulated database scan (accessing addresses $1, 2, \dots, N$) for 100 times
and found some parameters for which the stash size after every scan was zero.
(Such usage of scan is standard in literature since \cref{Path ORAM Lemma 3}
means scan maximizes the stash size.)
We emphasize that constructions with aggressive parameters lack rigorous
security and they are not suitable for fair comparison.

Unfortunately, we could not derive rigorous bounds for the second construction
(Theorem~\ref{second main theorem}) for reasonable size of $N$ since the
balls-into-bins analysis of Berenbrink et al.~\cite{Berenbrink00}, used in the
stash size analysis, requires a very large number of bins.
However, the simulation results indicate that the second construction works for
reasonable size of $N$.

\begin{table}
  \centering
  \caption{Performance comparison with concrete parameters. The symbol
  $\dagger$ means the integration of Ring ORAM techniques. $N=2^{20}$,
  $B=2^{10}$. $A$ and $S$ are parameters for the Ring ORAM. ($A$ specifies the
  infrequency of \textsf{EvictPath} and $S$ is the space in each bucket
  reserved for dummy blocks.) The cost for recursive calls and metadata
  handling are relatively minor and not included. The stash overflow
  probability is $<2^{-80}$ for rigorous settings. Aggressive settings do not
  have security guarantees (stash size bounds) and, in particular, are not
  suitable for fair comparison.}
  \label{table: performance with concrete params}
  \begin{tabular}{c c c | c c c}
    \multicolumn{2}{c}{} &
    \begin{tabular}[x]{@{}c@{}}Parameters\\$Z,L,M,A,S$\end{tabular} &
    \begin{tabular}[x]{@{}c@{}}Extra server\\space\end{tabular} &
    Bandwidth &
    Stash size \\

    \hline

    \multirow{4}{*}{\rotatebox{90}{Rigorous}} &
       \cite{Stefanov13}                       & 5,20,--,--,--  & $9N$      & 210  & 114 \\
    &  \cite{Ren15}                            & 5,19,--,4,6    & $10N$     & 109  & 63 \\
    &  Th.~\ref{main theorem}                  & 3,15,112,--,-- & $2.59N$   & 471  & 32 \\
    &  Th.~\ref{main theorem}$^\dagger$        & 5,15,112,4,7   & $2.91N$   & 253  & 64 \\

    \hline

    \multirow{6}{*}{\rotatebox{90}{Aggressive}} &
       \cite{Stefanov13}                       & 4,19,--,--,--  & $3N$      & 160  & \\
    &  \cite{Ren15}                            & 5,19,--,4,6    & $7N$      & 145  & \\
    &  Th.~\ref{main theorem}                  & 4,15,36,--,--  & $.25N$    & 288  & \\
    &  Th.~\ref{main theorem}$^\dagger$        & 5,15,36,4,6    & $.46875N$ & 163  & \\
    &  Th.~\ref{second main theorem}           & 3,16,14,--,--  & $.0625N$  & 248  & \\
    &  Th.~\ref{second main theorem}$^\dagger$ & 5,15,28,4,7    & $.25N$    & 194  & \\
  \end{tabular}
\end{table}

\section{Conclusion} \label{section: conclusion}

ORAM is a multifaceted problem and recently, researchers have been recognizing
the importance of rethinking the relevancy of multiple aspects of ORAM using
modern standards~\cite{Stefanov12, Bindschaedler15}.
In this paper, we provided another point of view and insight for this
exploration by introducing the notion of succinctness to ORAM and proposing
succinct ORAM constructions.
We think our methods are particularly suitable for secure processor setting.
It is interesting to consider succinct constructions optimized for other
settings.

As we already mentioned, we could not derive non-asymptotic bounds for
\cref{second main theorem} (for reasonable size of $N$) since the analysis of
Berenbrink et al.~\cite{Berenbrink00} for heavily loaded case of two-choice
balls-into-bins game, on which we rely, requires a large number of bins.
However, simulation results suggest that \cref{second main theorem} does have a
significant impact in practice even for reasonably small $N$ and it is
desirable to close this gap between theory and practice.
One obvious approach is to fine tune the analysis of Berenbrink et al. to get a
better (non-asymptotic) bound but this seems difficult because of the
complexity of the analysis.
Talwar and Wieder gave an alternative simpler analysis~\cite{Talwar14} but the
resulting bound is not as tight as that of Berenbrink et al.

\section*{Acknowledgement}

This work was supported by JSPS KAKENHI Grant Number 17H01693, 17K20023JST and
CREST Grant Number JPMJCR1402.
We thank Paul Sheridan for helpful discussion.

\bibliographystyle{plain}
\bibliography{succ_oram}

\end{document}